%% file: main.tex
\documentclass[11pt]{article}
\usepackage[margin=1.2in,bottom=1.2in]{geometry}
\usepackage[utf8]{inputenc}
\usepackage[T1]{fontenc}
\usepackage{libertinus}
\input{header}

\usepackage[
  natbib=true,
  style=numeric-comp,
  backend=biber,
  defernumbers,
  maxnames=99,
  maxcitenames=4
]{biblatex}

\usepackage{mathtools}

\usepackage{verbatim}

\usepackage{multicol}

\usepackage{tikz,pgfplots} 
\pgfplotsset{width=0.45\linewidth,compat=1.18} 

\usetikzlibrary{decorations}
\usetikzlibrary{patterns}
\usetikzlibrary{backgrounds}
\usepackage{circuitikz}
\usetikzlibrary{matrix}
\usetikzlibrary{graphs}
\usepackage{tikz-3dplot}
\usetikzlibrary{positioning, arrows.meta, calc, automata, fit}
\usetikzlibrary{decorations.pathreplacing, calligraphy}

\usepackage{xspace}

\title{
Computing the Determinant\\ via the Generalized Euclidean Algorithm} 

\newcommand{\aff}[1]{\textcolor{black!100}{\small{#1}}}
\author{Janina Reuter\\
   \aff{Kiel University} \\
   \aff{janina.reuter@informatik.uni-kiel.de}}

\date{}

\begin{document}
\maketitle

\begin{abstract}
We present an algorithm with a natural geometric interpretation for computing the determinant of a matrix $B\in\mathbb{Z}^{d\times d}$. It improves upon the current fastest deterministic algorithms by a factor of $d^{\omega(1)+1-\omega(2)}\approx d^{0.1213}$, where $\omega(k)$ denotes the exponent required for multiplying a $d\times d$ matrix with a $d\times d^k$ matrix. Our approach builds on a recent result of Klein and Reuter~(STOC~2025), who introduced a novel algorithmic idea for lattice basis computation that can be viewed as extending the Euclidean algorithm from $\mathbb{Z}$ to $\mathbb{Z}^d$. By adapting their techniques, we compute the determinant with the same bit complexity as applying the generalized Euclidean algorithm to an input matrix $A\in\mathbb{Z}^{d\times 2d}$ with $\|A\| = \|B\|$, namely $\tilde{O}(d^{\omega(2)}\log\|B\|)$. Prior to this work, the fastest deterministic algorithm for computing the determinant required $\tilde{O}(d^{\omega(1)+1}\log\|B\|)$ bit operations.
\end{abstract}

\section{Introduction}
The determinant is one of the most fundamental invariants of a matrix. 
A classical inequality by Hadamard states that $|\det(B)|\leq \prod_{i=1}^{d}\|\tup b_i\|_2$ for a matrix $B\in\zz^{d\times d}$, 
where equality holds if and only if the column vectors $\tup b_i$ are orthogonal.
In this case, the absolute value of the determinant equals the volume of the hyperrectangle spanned by the columns of $B$.
In this paper, we present an algorithm that admits a similarly natural geometric interpretation.
Our algorithm computes the determinant of any matrix, in particular orthogonality is not required and the sign of the determinant is computed as well instead of just the absolute value.

Roughly speaking, the algorithm computes for every column vector $\tup b_i$ a factor $z_i\in\zz_{>0}$ that resembles the contribution of this column to the volume of the parallelepiped 
\begin{align*}
    \Pi(B)=\{Bx\mid x\in[0,1)^d\}.
\end{align*} 
Similar to the hyperrectangle, the volume of the parallelepiped satisfies $\vol(\Pi(B)) = |\det (B)|$.
Compared to Hadamard's inequality,
the factors $z_i$ play a similar role as the lengths $\|\tup b_i\|_2$ of the column vectors in the case of orthogonal vectors.

We proceed to describe the core of the algorithm at a high level.
The algorithm computes a matrix $S \in\zz^{d\times d}$ with $\det(S)=1$ if the determinant of $B$ is positive and $\det(S) = -1$ if the determinant of $B$ is negative. 
Let $B^{(i)}$ denote the matrix, where the first $i$ columns of $B$ are changed to the first $i$ columns of $S$, i.e. $B^{(i)} = (\tup s_1, \ldots, \tup s_i, \tup b_{i+1}, \ldots \tup b_d)$.
Then, we track how the volume of the parallelepiped $B^{(i)}$ differs from the volume of the parallelepiped of $B^{(i-1)}$ by
\begin{align*}
    \vol(\Pi(B^{(i)})) \cdot z_i =  \vol(\Pi(B^{(i-1)})).
\end{align*}
The factors $z_i\in\zz_{>0}$, which we interpret as the contribution of $\tup b_i$ to the volume of $\Pi(B)$, yield the absolute value of the determinant with
\begin{align*}
    |\det(B)| = \vol(\Pi(B)) 
    = \prod_{i=1}^dz_i,
\end{align*}
since $B^{(d)} = S$ and $|\det(S)| = 1$. 
Additionally, we can determine the sign of $\det(B)$ by $\det(S)$.

Besides this conceptually simple procedure, the algorithm also improves upon the complexity of the current fastest deterministic algorithms.
Regarding the number of arithmetic operations, computing the determinant of a matrix reduces to matrix multiplication.
As with many linear algebraic problems, algorithms for determinant computation can suffer from \emph{intermediate expression swell}. 
This is a phenomenon, where during computation the bitlength of numbers might get exponentially larger than the bitlength of input numbers and numbers in the solution. 
In practice, computations are performed using machine words of fixed size and thus a single arithmetic operation on such numbers requires an exponential number of word operations. 
Hence, the extensive body of research on determinant computation has primarily focused the more precise measure of bit complexity, that is, counting the required number of operations on bits. 
While existing techniques can prevent exponential growth, the dependence on the dimension of the state-of-the-art bit complexity still exceeds arithmetic complexity by a linear factor leading to an exponent of $\omega +1$, where $\omega$ denotes the matrix multiplication exponent.

The present algorithm builds upon a recent result by Klein and Reuter~\cite{DBLP:conf/stoc/KleinR25}, which considers another problem from linear algebra and previously had a similar bit complexity. 
They showed how the influence of large intermediate numbers in their algorithm is heavily restricted with a notable difference only in a single step during the setup before the main part of the procedure. 
Within this step they avoid large entries of matrices during multiplication by spreading large numbers over linearly many columns. 
The resulting blow-up in dimension is then handled with rectangular matrix multiplication. 
The exponent for rectangular matrix multiplication on dimensions $d\times d$ and $d\times d^2$, denoted by $\omega(2)$, (currently) exceeds $\omega$ by strictly less than $1$~\cite{DBLP:conf/soda/AlmanDWXXZ25}.
Hence, by restricting the need to handle large numbers to an operation that is equivalent to a single rectangular matrix multiplication, their speed-up is $\omega + 1 - \omega(2) \approx 0.1213$ in the exponent of the matrix dimension. 
In this work we show how to modify their algorithm to obtain a similar speed-up for the problem of computing the determinant of a matrix.

\subsection{Related Work}

There is a vast literature on computing the determinant of a matrix~\cite{modernComputerAlgebra,DBLP:journals/cc/KaltofenV05,DBLP:journals/jc/Storjohann05,DBLP:conf/issac/PauderisS12, DBLP:conf/issac/Kaltofen02,DBLP:conf/issac/Kaltofen92,DBLP:journals/jacm/TzameretC21,DBLP:conf/focs/EberlyGV00, DBLP:conf/issac/AbbottBM99}. For an overview of techniques, we refer to a survey by Kaltofen and Villard~\cite{kaltofen2004computing}.

The state-of-the-art bit complexity can be achieved as a combination of two standard techniques. The first part is to compute a triangular factorization of the input matrix $B\in \zz^{d\times d}$, that is, a unit lower triangular matrix $L$, an upper triangular matrix $U$, and permutation matrix $P$ such that $M=L\,UP$. The unit lower triangular matrix has $\det(L) = 1$ and hence $\det(B) = \det(P)\det(U)$, which can be computed by the diagonal entries of $U$ and number of exchanged rows in $P$. 
A triangular factorization can be computed using $\tilde{O}(d^{\omega})$ arithmetic operations, for example with the algorithms by Bunch and Hopcroft~\cite{bunch1974triangular} or Hafner and McCurley~\cite{DBLP:journals/siamcomp/HafnerM91}. 
Due to intermediate expression swell, it is not immediately clear that such procedures have polynomial bit complexity, as illustrated for example by von zur Gathen and Gerhard~\cite[Chapter~5.5]{modernComputerAlgebra}. 

The second part is thus to cope with potentially large intermediate numbers. 
The usual path here uses modular arithmetic regarding several pairwise coprime moduli that multiply together to be of an order similar to the size of numbers in the solution. Chinese remaindering then reconstructs the solution from the results with respect to the small moduli.
Using Hadamard's bound on determinants, multiplied moduli of order $(d\log\|B\|)^d$ are sufficient, where $\|.\|$ denotes the infinity norm.
For an overview of modular arithmetic, we refer to book of von zur Gathen and Gerhard~\cite[Chapter~5]{modernComputerAlgebra}. 
Combining fast triangulation with modular arithmetic yields the currently best bit complexity of
\begin{align*}
    \tilde{O}(d^{\omega+1}\log\|B\|).
\end{align*}
Only in the special case of testing unimodularity, so whether the determinant of a matrix is $-1$ or $1$, an algorithm with an improved complexity of $\OTilde(d^{\omega}\log{\norm{B}})$ exists as shown by Paudris and Storjohann~\cite{DBLP:conf/issac/PauderisS12}.

Complexity bounds of the form \texttt{arithmetic complexity}$\,\times\,\log(\mathrm{\texttt{Hadamard bound}})$ seem to be a common barrier for linear algebraic problems related to matrices, with  numbers bounded by determinants and hence Hadamard's inequality. Examples for this, besides triangularization and determinant computation, are the Smith normal form~\cite{DBLP:conf/issac/Storjohann96}, Hermite normal form~\cite{DBLP:conf/issac/StorjohannL96}, and the inverse of a nonsingular matrix~\cite[Chapter~6.3]{DBLP:books/aw/AhoHU74},~\cite[Chapter~5]{modernComputerAlgebra}. 
Up until recently this held true also for solving a linear system of equations $B\tup x=\tup b$ for a nonsingular matrix $B\in\zz^{d\times d}$ and right-hand side vector $\tup b\in\zz^d$ with a bit complexity of $\tilde{O}(d^{\omega +1}\log\|B\|)$ if $\log\|\tup b\|\in O(d\log(d\|B\|))$.
The arithmetic complexity is equal to matrix multiplication because computing the inverse reduces to matrix multiplication~\cite[Chapter~6.3]{DBLP:books/aw/AhoHU74} and $\tup x=B^{-1}\tup b$. Moreover, by Cramer's rule the entries of the solution vector $\tup x$ are given by the determinant of $B$ and determinants of matrices with the a column of $B$  changed to $\tup b$ and are therefore bounded by Hadamard's inequality.
A surprising result by Birmpilis, Labahn, and Storjohann~\cite{DBLP:conf/issac/BirmpilisLS19} shows how to solve linear systems with bit complexity $\tilde{O}(d^\omega\log\|B\|)$. To a certain extend, their result derandomizes a Las Vegas algorithm with the same bit complexity~\cite{DBLP:journals/jc/Storjohann05}. Randomized algorithms of Las Vegas type have expected running time, while providing an output that is certified to be correct. 

In contrast to deterministic computations, there exist faster randomized algorithms for most of the linear algebraic problems mentioned above~\cite{DBLP:journals/jc/Storjohann05,DBLP:journals/cc/Storjohann15,DBLP:journals/corr/abs-2605-07784, DBLP:conf/issac/Storjohann96}. Regarding the determinant, the fastest (Las Vegas) randomized algorithm is by 
Storjohann~\cite{DBLP:journals/jc/Storjohann05} with an expected bit complexity of $\tilde{O}(d^{\omega}\log\|B\|)$. Hence, a natural open question is whether randomization is required to achieve this bit complexity.

\subsection{The Generalized Euclidean Algorithm for Lattice Basis Computation}

The \emph{lattice} generated by the column vectors of a matrix $A = (\tup a_1, \ldots, \tup a_n) \in \zz^{d\times n}$, denoted $\lattice(A)$, is the set
\begin{align*}
    \lattice(A) = \{A\tup x \mid \tup x\in\zz^n\}.
\end{align*}
A basic fact from lattice theory is that every lattice can be generated by linearly independent vectors. 
The problem of \emph{lattice basis computation} considers an input matrix $A\in\zz^{d\times n}$ and asks for a set of linearly independent vectors $S\subset \zz^d$ with $\lattice(S) = \lattice(A)$. 

Very recently, Klein and Reuter~\cite{DBLP:conf/stoc/KleinR25} introduced a novel algorithmic idea for lattice basis computation.
While previously this problem was usually solved via matrix normal forms, such as the Hermite normal form, their algorithm is a direct geometric computation of elements from $S$. 
For a smooth presentation, we assume that the input matrix has full row rank and allow set operations on matrices defined as follows. We write $B\setminus\tup b_i$ for deleting column $i$ of matrix $B$ and we write  $B\setminus\tup b_i\cup \tup c$ for changing column $i$ of $B$ to vector $\tup c$. For matrices $A=(\tup a_1,\ldots,\tup a_n)$ and $B=(\tup b_1,\ldots,\tup b_m)$ with the same number of rows, we denote with $A\cup B$ the matrix with the columns by $A$ followed by the columns by $B$, i.e. $A\cup B = (\tup a_1,\ldots,\tup a_n,\tup b_1,\ldots,\tup b_m)$.

Roughly speaking, the algorithm on input $A\in \zz^{d\times n}$ considers a \emph{starting basis} $B\in\zz^{d\times d}$, which is not necessarily a basis of $\lattice(A)$, consisting of linearly independent vectors in $A$ and \emph{additional vectors} $C = A\setminus B$. The central idea can be summarized by the following question.
\begin{quote}
    Which is the closest affine translate of $\mathrm{span}(B\setminus \tup b_i)$ containing an element of $\lattice(A)$?
\end{quote}
They show that the lattice $\lattice(A)$ is then generated by any vector $\tup s_i$ on this closest translate together with any set of vectors generating the lattice restricted to the subspace $\mathrm{span}(B\setminus \tup b_i)$. This procedure can be iterated reducing the dimension of the considered subspace by one in each iteration. 
Following this approach, the algorithm computes linearly independent vectors $S=(\tup s_1, \ldots, \tup s_d)$ that generate the lattice. See also~\autoref{fig:gcd_translate} for an example in dimension $d=2$.

\begin{figure}
    \centering
    \includegraphics[width=0.85\linewidth]{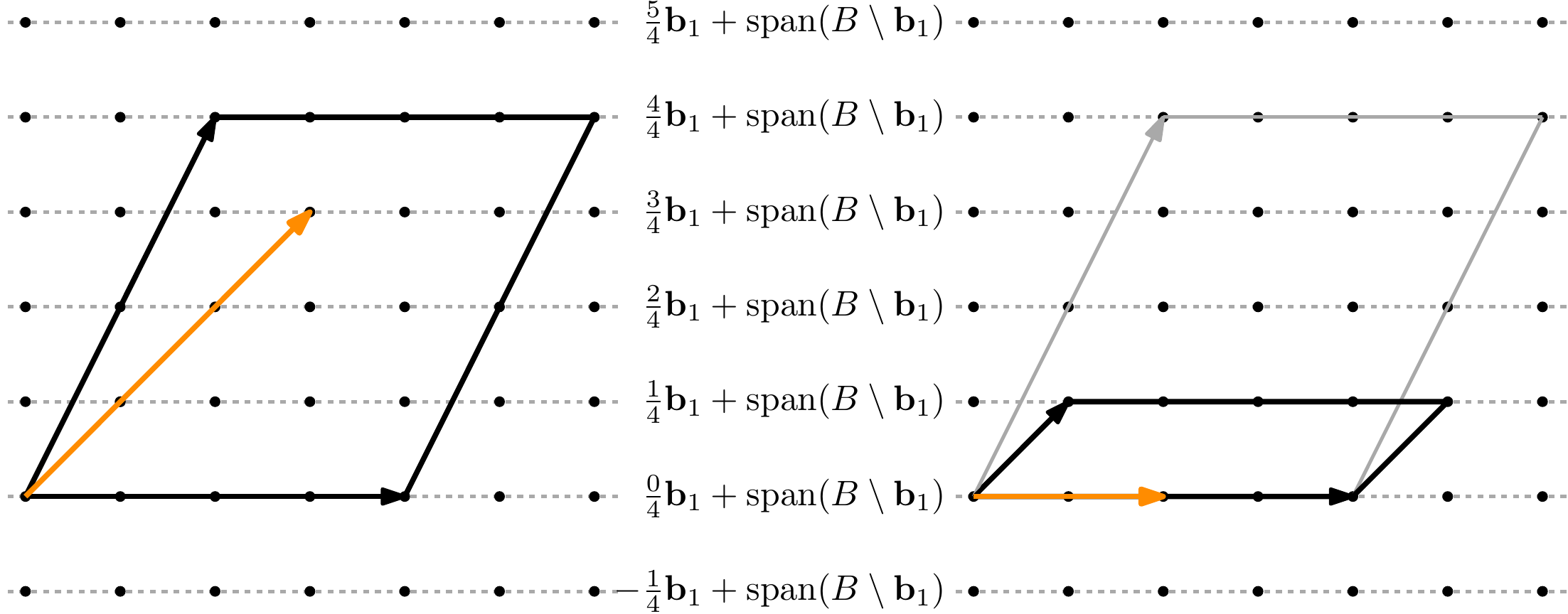}
    \caption{An example of the generalized Euclidean algorithm on input 
    $\tup a_1=(2,4)^\intercal$, $\tup a_2=(4,0)^\intercal$, and $\tup a_3=(3,3)^\intercal$ with starting basis $\tup b_1 = \tup a_1$ and $\tup b_2=\tup a_2$. Vector $\tup c=\tup a_3 = \frac 3 4 \tup b_1 + \frac38\tup b_2$ lies on the third affine subspace reachable by integral combinations of $\tup b_1$ and $\tup c$. The closest reachable subspace is $\frac14\tup b_1 + \mathrm{span}(B\setminus\tup b_1)$. The algorithm finds a lattice point $\tup s_1 =(1,1)^\intercal$ on this subspace and projects $\tup c$ to a lattice point $(2,0)^\intercal$ of the subspace $\mathrm{span}(B\setminus\tup b_1)$. With another iteration on the lower dimensional subspace, the algorithm finds the second vector $\tup s_2 = (2,0)^\intercal$ of the solution basis. 
    }
    \label{fig:gcd_translate}
\end{figure}

\begin{algorithm}
  \caption{Generalized Euclidean Algorithm
    \label{alg:fasteuclidean}}
  \begin{algorithmic}[1]
    \Statex{
    \textsc{Input: } $A=\left(\tup a_1,\ldots, \tup a_n \right) \in \zz^{\dd\times n}$ }
    \State{find linearly independent vectors $B = \left(\tup b_1, \ldots, \tup b_\dd\right)$ where $\tup b_i\in \{\tup a_1,\ldots, \tup a_n\}$}
    \State{let $C$ be a matrix with columns $\{\tup a_1,\ldots, \tup a_n\}\setminus \{\tup b_1, \ldots, \tup b_\dd\}$}
    \State{solve $B X = C$}
    \For{$i=1$ to $\dd$}
        \State{find a row of $X$ maximizing the lcm of denominators and w.l.o.g. assume this is row $i$}
        \State{represent row $i$ of $X$ as $x_{i\tup c} = \frac{y_{\tup c}}{z_i}$, where $z_i$ is the lcm of denominators in row $i$}
        \State{consider the translate numbers $y_{\tup c}$ for every $\tup c\in C$ and $z_i$ for $\tup b_i$}
        \State{find a vector $\tup s_i \in \lattice(A)$ which lies in the subspace $\frac{1}{z_i}\tup b_i + \mathrm{span}(B\setminus \tup b_i)$ using the}
        \Statex{\phantom{\textbf{for}} extended Euclidean algorithm (in dimension one)}
        \State{add $\tup s_i$ to the solution basis $S$}
        \State{project the vectors in $C$ to the subspace $\mathrm{span}(B\setminus \tup b_i)$ and update $X$ accordingly}
    \EndFor
    \State \Return{basis $S$}
  \end{algorithmic}
\end{algorithm}

\autoref{alg:fasteuclidean} reformulates their algorithm aligning with our intended purpose. In order to achieve their bit complexity, the algorithm additionally maintains the invariant that $\|X\|\leq 1$ by reducing fractional entries modulo $1$ and integral entries to $1$. We keep line 10 informal here as this work does not rely on the specifics of this step. It is, however, a simple linear combination using the multipliers from the extended Euclidean algorithm as computed in line 8 and can be found in~\cite{DBLP:conf/stoc/KleinR25}.

\begin{theorem}[\cite{DBLP:conf/stoc/KleinR25}]\label{theo:bit_complexity_generalized_euclid}
    \autoref{alg:fasteuclidean} computes a basis of the lattice $\lattice(A)$. The bit complexity is bounded by
    $\tilde{O}(\max\{n-d, d\}\,d^{\omega(2)-1}\log\|A\|)$. 
\end{theorem}

\subsection{Our Contribution}

This work shows a conceptually simple way to compute the determinant of a matrix $B\in \zz^{d\times d}$ using $\tilde{O}(d^{\omega(2)}\log\|B\|)$ bit operations by extending the algorithm of Klein and Reuter~\cite{DBLP:conf/stoc/KleinR25}. This improves upon the long standing 
bound of $\tilde{O}(d^{\omega+1}\log\|A\|)$ by a factor of $\approx d^{0.1213}$ using the current values of $\omega$ and $\omega(2)$ by Alman et al.~\cite{DBLP:conf/soda/AlmanDWXXZ25}.
Modifying the algorithm requires \emph{three and a half} insights structuring this paper as follows. 
\begin{itemize}
    \setlength\itemsep{0.25em}
    \item Check whether the input is singular, hence returning determinant $0$, see~\autoref{sec:singularity}.
    \item Compute the absolute value of the determinant, see~\autoref{sec:absolute_value}.
    \item Compute the sign of the determinant using the computed basis $S$, see~\autoref{sec:sign}.
    \item Combining sign and absolute value of the determinant then yields the determinant as discussed in the analysis in~\autoref{sec:analysis}.
\end{itemize}

\section{Computing the Determinant}

Our algorithm uses the generalized Euclidean algorithm as stated in~\autoref{alg:fasteuclidean} as a template with only few modifications.  
The modified algorithm for computing the determinant is stated in the final subsection in~\autoref{alg:determinant} after discussing the required changes. 
The first of which is to directly input $B\in\zz^{d\times d}$ instead of line 1 and a fixed matrix of $C=I_d$ in line 2. The goal is then to compute the determinant of $B$ using the fact that the lattice generated by $\lattice(B\cup C) = \lattice (B\cup I_d)$ is the integer lattice $\zz^d$ as the identity matrix generates all integral points. 

Before discussing the procedure in detail, we reiterate the core of our geometric intuition applied to the generalized Euclidean algorithm. Consider again the example in~\autoref{fig:gcd_translate}. The absolute value of the determinant of any nonsingular matrix $B\in\zz^{d\times d}$ is the volume of its parallelepiped $\Pi(B) = \{Bx\mid x\in[0,1)^d\}$. Changing $\tup b_1$ to $\tup s_1$ reduces the volume by a factor of $\frac 14$ because the parallepiped $\Pi(B\setminus \tup b_1 \cup \tup s_1)$ spans only $\frac14$ into the direction of $\tup b_1$. We use this idea to relate the volume of $\Pi(B)$ to the volume of $\Pi(S)$. With $C=I_d$, the latter is $1$ because any basis of the lattice $\zz^d$ has determinant $-1$ or $1$. Hence, we get the absolute value of the determinant by tracking the translate numbers imposing a change of the parallelepiped's volume. Together with a singularity test and computation of the sign, this yields the determinant.

\subsection{Deciding Singularity}\label{sec:singularity}

We use the following results from~\cite{DBLP:conf/stoc/KleinR25} to decide singularity of a square matrix. While the original results consider rectangular matrix multiplication for dimension $d\times d$ and $d\times d^k$ for arbitrary $k\geq 1$, our algorithm is optimal for $k=2$ and hence for simplicity we restate the results only for this case.

\begin{lemma}[\cite{DBLP:conf/stoc/KleinR25}]\label{lem:MM}
    Consider two matrices $M\in \zz^{d \times d}$ and $N\in \zz^{d \times d}$, where $\log\|N\|\in \tilde O(d\log\|M\|)$. Then the bit complexity of the matrix multiplication $M N$ is bounded by
    \[\OTilde(d^{\omega(2)}\log\|M\|).\]
\end{lemma}
\begin{lemma}[\cite{DBLP:journals/jsc/BirmpilisLS23,DBLP:conf/stoc/KleinR25}]\label{lem:lin_sys}
    Given a nonsingular matrix $M \in \zz^{d \times d}$ and a right hand-side matrix $R \in \zz^{d \times d}$, where $\log\norm{R} \in \mathcal O(d\log\norm{M})$.
    Then a solution $X\in\mathbb{Q}^{d\times d}$ to the linear system $MX=R$ can be found in bit complexity
    \begin{align*}
        \OTilde(d^{\omega(2)}\log\norm{M}).
    \end{align*}
\end{lemma}

It seems worth noting that~\autoref{lem:lin_sys} implies an improved complexity for computing the inverse of a nonsingular matrix $B\in\zz^{d\times d}$. This has
not been mentioned explicitly yet, hence it feels justified to do so here.
To the best of our knowledge, the fastest algorithm thus far required $\tilde{O}(d^{\omega+1}\log\|B\|)$ bit operations by combining a reduction to matrix multiplication~\cite[Chapter~6.3]{DBLP:books/aw/AhoHU74} with the classical approach of modular arithmetic to reduce intermediate numbers~\cite[Chapter~5]{modernComputerAlgebra}. 
In 2015, Storjohann showed a Las Vegas randomized algorithm involving the condition number $\kappa(B) = d\,\|B\|\,\|B^{-1}\|$ with a complexity of $\tilde{O}(d^3(\log\|B\|+\log\|\kappa(B)\|))$. 
Hence, for well-conditioned matrices, where $\kappa(B)$ is a polynomial of $n\log\|B\|$, this algorithm yields a cubic bit complexity. However, in general, the condition number can be as large as $\log(\kappa(B))\in \tilde{O}(d\log\|B\|)$ and thus this algorithm requires up to $\tilde{O}(d^4\log\|B\||))$ bit operations. 

The above result for linear system solving with matrix $M=B$ and the identity matrix as right-hand side, so $R=I_d$, yields a general improvement for inverse computation. For the current values of $\omega<2.371339$ and $\omega(2)<3.250035$~\cite{DBLP:conf/soda/AlmanDWXXZ25}, this closes the gap between the exponent of the previous complexity bound ($\omega +1 < 3.371339$) and the cubic space complexity by roughly one third.

\begin{corollary}[\cite{DBLP:conf/stoc/KleinR25}]
     The inverse of a nonsingular matrix $B\in\zz^{d\times d}$ can be computed with bit complexity  $\tilde{O}(d^{\omega(2)}\log\|B\|)$.
\end{corollary}

We try to run the algorithm behind~\autoref{lem:lin_sys} in an attempt to compute the inverse of the possibly singular input matrix $B$. Independent of how their algorithm works exactly and independent of the fact that it is stated for nonsingular matrices we can deduct whether the input matrix is singular.

\begin{algorithm}
  \caption{Test Singularity
    \label{alg:singularity}}
  \begin{algorithmic}[1]
    \Statex{
    \textsc{Input: } A matrix $B=\left(B_1,\ldots, B_d \right) \in \zz^{\dd\times \dd}$ }
    \State{Try to compute $BX=I_d$ using \autoref{lem:lin_sys}}
    \State{If that failed, return true}
    \State{If that succeeded, retrieve $X$}
    \State{If $BX=I_d$ return false, else return true}
  \end{algorithmic}
\end{algorithm}

\begin{lemma}\label{lem:singularity}
    \autoref{alg:singularity} correctly determines whether the input matrix $B\in  \zz^{d\times d}$ is singular in bit complexity $\tilde{O}(d^{\omega(2)}\log\norm{B})$.
\end{lemma}
\begin{proof}
    Assume the matrix $B$ is nonsingular. Then the algorithm behind~\autoref{lem:lin_sys} correctly determines $X=B^{-1}$ in target time. The matrix multiplication in line 4 can be computed in target time because numerators and the least common multiple (lcm) of denominators of $X$ are both bounded by $(d\norm{B})^d$ by Cramer's rule and Hadamard's bound on determinants and hence multiplication by the lcm of denominators, matrix multiplication using~\autoref{lem:MM}, and finally division by the lcm of denominators yields the result of the matrix multiplication in target time.
    By correctness of~\autoref{lem:lin_sys}, false is returned.
    
    If otherwise the matrix $B$ is singular, \autoref{alg:singularity} returns this fact if the algorithm, that is designed for nonsingular matrices, at some point fails to continue its computation. If otherwise the algorithm on unintended input of a singular matrix $B$ does return some $X$, then \autoref{alg:singularity} checks its correctness. For a singular $B$ there does not exist $X$ such that $BX=I_d$ and hence in this case $BX\neq I_d$ for any (false) output $X$ from line 2. The matrix multiplication $BX$ can be done in target time similar to the first case. In particular, any result of~\autoref{lem:lin_sys} has suitably bounded entries due to the underlying modulo calculations. In this case, the algorithm returns true.    
\end{proof}

\subsection{Computing the Absolute Value}\label{sec:absolute_value}

Let $B^{(i)}$ denote the input matrix $B$, where the first $i$ columns are changed to the vectors $\tup s_i$ on the gcd of translates computed in iteration $i$ of \autoref{alg:fasteuclidean} lines 4-10, i.e. $B^{(i)} = (\tup s_1, \ldots \tup s_{i}, \tup b_{i+1}, \ldots \tup b_d)$. In particular, we have $B^{(0)}=B$ and $B^{(d)} = S$. 
With the reformulation of the generalized Euclidean algorithm as stated in~\autoref{alg:fasteuclidean}, a direct implication on changes of the determinant from $B^{(i)}$ to $B^{(i+1)}$ follows from Cramer's rule.

\begin{fact}[Cramer's rule]
    Consider a nonsingular matrix $A\in\zz^{d\times d}$ and right-hand side $\tup b\in \zz^d$.
    Then the vector $x = A^{-1}\tup b$ has individual entries $x_i = \frac{\det(A\setminus \tup a_i\cup \tup b)}{\det(A)}$.
\end{fact}

By line 8 of the pseudcode, we have that $x_i = \frac 1 {z_i}$ for $B^{(0)}\tup x =\tup s_{i}$. 
Line 10 projects vectors of $C$ the subspace $\mathrm{span}(B\setminus\tup b_i)$. This implies that $x_1 = \ldots = x_{i-1} = 0$ and hence $B^{(0)}\tup x = B^{(i-1)}\tup x$ because the matrices only differ in columns that are not used by $\tup x$.
Cramer's rule implies 
\begin{align*}
    \frac{1}{z_i} =  \frac{\det(B^{(i)})}{\det(B^{i-1})}.
\end{align*}
Iterating this argument we get the following lemma.

\begin{lemma}\label{lem:absolute_value_of_det}
    Consider a nonsingular matrix $B\in \zz^{d\times d}$ as computed in line 1, denominators $z_i\in\zz_{\geq 1}$ computed in line 6, and $S\in\zz^{d\times d}$ the returned basis. 
    Then $\det(B^{(i-1)}) = z_i \cdot \det(B^{(i)})$ holds for every $i\leq d$ and iterating this argument yields
    \begin{align*}
        \det(B)  = \det(S)\, \prod_{i=1}^d z_i. 
    \end{align*}
\end{lemma}

Assuming the generated lattice is $\zz^d$, which is the case for example if $C=I_d$ in line 3, we get a direct implication on the absolute value of the determinant of $B$ as in this case $|\det(S)|=1$. Moreover, the determinants of $B$ and $S$ have the same sign since all denominators $z_i$ are positive. 

\begin{corollary}\label{obs:abs_det}
    If the generated lattice is $\lattice(A) = \zz^d$, then the absolute value of the determinant of $B$ as computed in line 1 is
    \begin{align*}
        |\det(B)|  = \prod_{i=1}^d z_i. 
    \end{align*}
\end{corollary}

\begin{corollary}\label{obs:same_sign_dets}
    For every $B$ computed in line 1 of~\autoref{alg:fasteuclidean} and returned basis $S$, it holds that $\sgn(\det(B)) = \sgn(\det(S))$.
\end{corollary}

\subsection{Computing the Sign}\label{sec:sign}

This subsection combines~\autoref{obs:same_sign_dets} with a standard technique for computation of small determinants.
We repeat the arguments in the following lemma for the sake of completeness. 

\begin{lemma}\label{lem:det_mod}
    Given a matrix $M\in \zz^{d\times d}$ with $|\det(M)|\leq \Delta$, the determinant can be computed in time $\tilde{O}(d^{\omega}\log\Delta + d^2\log{\|M\|})$. 
\end{lemma}
\begin{proof}
    We compute the determinant in $\zz_{2\Delta +1}$, where the entry-wise modulo on $M$ is computed in $\tilde{O}(d^2\log{\|M\|})$.
    This can be done for example using any fast algorithm for triangular matrix factorization in $\zz_{2\Delta +1}$. 
    There are several decompositions available for this operation to be performed in $\tilde{O}(d^{\omega}\log\Delta)$, 
    see for example~\cite{DBLP:journals/siamcomp/HafnerM91}.

    If the determinant is non-negative, then $\det(M) \in\zz_{2\Delta+1}$ implies the claim. 
    If the determinant is negative, then we have $-\Delta \leq \det(M) < 0$ and thus $\det(M) + (2\Delta+1) \in \zz_{2\Delta+1}$. In other words, 
    given a result of $\det(M)\bmod (2\Delta+1) > \Delta$, we can deduce $\det(M)<0$ with 
    
    \begin{align*}
        \det(M) = (\det(M)\bmod (2\Delta+1)) - (2\Delta +1).
    \end{align*}
    \vspace{-2\baselineskip} \\
    \phantom{.}
\end{proof}

In the next subsection, we want to apply the lemma above to the computed basis $S$. With fixed $C=I_d$ in line 3 of~\autoref{alg:fasteuclidean}, the resulting lattice is $\zz^d$ and thus $S$ is unimodular, i.e. $|\det(S)|=1$. By~\autoref{obs:same_sign_dets} we get a positive sign of $\det(B)$ if $\det(S)=1$ and a negative sign if $\det(S)=-1$. 
\begin{corollary}\label{cor:sign_unomodular_complexity}
    For a unimodular matrix $U\in\zz^{d\times d}$ the determinant can be computed with bit complexity $\tilde{O}(d^{\omega} + d^2\log\|U\|)$.
\end{corollary}

\subsection{Analysis}\label{sec:analysis}

\autoref{alg:determinant} implements the ideas of the previous subsections yielding our main result as stated in the following theorem.

\begin{algorithm}
  \caption{Determinant Computation via the Generalized Euclidean Algorithm
    \label{alg:determinant}}
  \begin{algorithmic}[1]
    \Statex{
    \textsc{Input: } $B=\left(B_1,\ldots, B_d \right) \in \zz^{\dd\times \dd}$ }
    \State{let $C = I_d$ be the identity matrix}
    \State{test singularity of $B$ using \autoref{alg:singularity}, if true is returned, return $0$}
    \State{continue with $X=B^{-1}$ as computed in \autoref{alg:singularity}}
    \For{$i=1$ to $\dd$}
        \State{find a row of $X$ maximizing the lcm of denominators and w.l.o.g. assume this is row $i$}
        \State{represent row $i$ of $X$ as $x_{i\tup c} = \frac{y_{\tup c}}{z_i}$, where $z_i$ is the lcm of denominators in row $i$}
        \State{consider the translate numbers $y_{\tup c}$ for every $\tup c\in C$ and $z_i$ for $\tup b_i$}
        \State{find a vector $\tup s_i \in \lattice(A)$, which lies in the subspace $\frac{1}{z_i}\tup b_i + \mathrm{span}(B\setminus \tup b_i)$ using the}
        \Statex{\phantom{\textbf{for}} extended Euclidean algorithm (in dimension one)}
        \State{add $\tup s_i$ to the solution basis $S$}
        \State{project the vectors in $C$ to the subspace $\mathrm{span}(B\setminus \tup b_i)$ and update $X$ accordingly}
    \EndFor
    \State{compute $\det(S)$ using \autoref{cor:sign_unomodular_complexity}}
    \State \Return{$\det(S)\cdot \prod_{i=1}^dz_i$}
  \end{algorithmic}
\end{algorithm}

\begin{theorem}
    \autoref{alg:determinant} computes the determinant of a matrix $B\in\zz^{d\times d}$ in bit complexity 
    \begin{align*}
        \tilde{O}(d^{\omega(2)}\log\|B\|).
    \end{align*}
\end{theorem}
\begin{proof}
    Let us first consider correctness.
    If $B$ is singular, then by~\autoref{lem:singularity} line 2 detects and returns this fact in target time. 
    If otherwise $B$ is nonsingular, then \autoref{alg:determinant} correctly computes the determinant of $B$ by \autoref{obs:abs_det} and \autoref{obs:same_sign_dets} because $\lattice(B\cup C) = \lattice (B\cup I_d) = \zz^d$ and thus the computed basis $S$ is  unimodular.

    Regarding bit complexity, testing singularity and computing the sign of $\det(S)$ are in target time by~\autoref{lem:singularity} and~\autoref{cor:sign_unomodular_complexity}. The same is true for the multiplication in line 12. 
    Apart from these operations, \autoref{alg:determinant} only performs (most parts of) the generalized Euclidean algorithm,\autoref{alg:fasteuclidean}, on input $A\in\zz^{d\times 2d}$ consisting of the columns of $B$ and $I_d$, which is in target time by~\autoref{theo:bit_complexity_generalized_euclid}. In total, the bit complexity is bounded by $\tilde{O}(d^{\omega(2)}\log\|B\|)$.
\end{proof}

\printbibliography

\end{document}

%% file: header.tex
\usepackage{csquotes}
\usepackage{microtype}
\usepackage{mathtools}
\usepackage{amssymb}
\usepackage{amsmath}
\usepackage{amsthm}
\usepackage{thmtools}
\usepackage{thm-restate}
\usepackage{hyperref}
\usepackage{graphicx}
\usepackage{enumerate}  
\usepackage{multirow}
\usepackage{braket} 
\usepackage{tcolorbox}%

\usepackage{tikz}
\usepackage{circuitikz}
\usepackage{enumitem}
\usepackage{svg}

\usepackage{algorithmicx}%
\usepackage[noend]{algpseudocode}%
\usepackage{algorithm}

\newtheorem{theorem}{Theorem}

\newtheorem{lemma}{Lemma}
\newtheorem{fact}{Fact}
\newtheorem{corollary}{Corollary}

\newtheorem{problem}{Problem}

\theoremstyle{remark}

\newcommand{\zz}{\mathbb{Z}}

\newcommand{\dd}{d}

\newcommand{\tup}[1]{\mathbf{#1}}

\newcommand{\norm}[1]{\left\lVert #1 \right\rVert}

\DeclareMathOperator{\lattice}{\mathcal{L}}

\DeclareMathOperator{\vol}{vol}
\DeclareMathOperator{\sgn}{sgn}

\newcommand{\OTilde}{\ensuremath{\Tilde{\mathcal{O}}}}